\documentclass[11pt,oneside,letterpaper]{article}
\usepackage{fullpage}
\usepackage{times,fancyhdr}
\usepackage[dvips]{graphicx}
\usepackage{subfigure}
\usepackage{amsmath, amsfonts, amssymb, amsthm}
\usepackage{pstricks}
\usepackage{hyperref}
\usepackage{mathrsfs}
\usepackage{algorithm,algorithmicx,algpseudocode}
\usepackage{comment}

\theoremstyle{definition}
\newtheorem{Thm}{Theorem}[section]
\newtheorem{Lem}[Thm]{Lemma}
\newtheorem{Cor}[Thm]{Corollary}

\newtheorem{Def}[Thm]{Definition}

\newtheorem{Prop}[Thm]{Proposition}
\newtheorem{Obse}[Thm]{Observation}

\newcommand{\concept}[1]{\emph{#1}}

\begin{document}

\title{Approximate Capacities of Two-Dimensional Codes\\ by Spatial Mixing}
\author{Yi-Kai Wang\thanks{State Key Laboratory for Novel Software Technology, Nanjing University, China.}~\thanks{Supported by NSFC grants 61272081 and 61321491. Email: \texttt{yikai.wang@hotmail.com}.}
  \and Yitong Yin\footnotemark[1]~\thanks{Supported by NSFC grants 61272081 and 61321491. Email: \texttt{yinyt@nju.edu.cn}.}
  \and Sheng Zhong\footnotemark[1]~\thanks{Supported by RPGE, NSFC-61321491, and NSFC-61300235. Email: \texttt{zhongsheng@nju.edu.cn}.}}

\date{}

\maketitle

\begin{abstract}
We apply several techniques developed in recent years for counting algorithms and statistical physics to study the \emph{spatial mixing} property of two-dimensional codes arising from local hard (independent set) constraints, including: hard-square, hard-hexagon, read/write isolated memory (RWIM), and non-attacking kings (NAK). For these constraints, the existence of strong spatial mixing implies the existence of polynomial-time approximation scheme (PTAS) for computing the capacity.
The existence of strong spatial mixing and PTAS was previously known for the hard-square constraint.
We show the existence of strong spatial mixing for hard-hexagon and RWIM constraints by establishing the strong spatial mixing along self-avoiding walks, which implies PTAS for computing the capacities of these codes.  We also show that for the NAK constraint, the strong spatial mixing does not hold along self-avoiding walks.
\end{abstract}

\section{Introduction}

We consider codes consisting of two-dimensional binary patterns of 0's and 1's arranged in rectangles, satisfying constraints of forbidding certain local patterns.
The capacity (or entropy) of a two-dimensional code measures the maximum rate that the information can be transmitted through this representation. Computation of capacities of two-dimensional codes is highly nontrivial and has implications in information theory~\cite{baxter,calkin,weeks,kato,nagy,cohn,golin,sabato,pavlov,marcus}, probability~\cite{pavlov,marcus}, statistical physics~\cite{gamarnik2009sequential} and computation theory~\cite{hochman}.
A recent breakthrough~\cite{hochman} in symbolic dynamics shows that capacities of two-dimensional codes  characterize the class of Turing-computable reals. It is then a natural and fundamental problem to further classify the capacities of two-dimensional codes which can be computed \emph{efficiently}, that is, in polynomial time.


In this paper, we consider the class of two-dimensional codes which: (1) can be described by local (up-to distance 2) forbidden patterns, and (2) arise from hard (independent set) constraints. This gives us precisely the following well-studied constraints for two-dimensional codes: hard square (HS)~\cite{calkin,pavlov}, hard hexagon (HH)~\cite{baxter}, read/write isolated memory (RWIM)~\cite{cohn,golin} and non-attacking kings (NAK)~\cite{weeks}. Among these constraints, the capacities of hard-square and hard-hexagon are known to be efficiently computable~\cite{baxter,pavlov,marcus}. The efficiency of computation of the capacity of RWIM or NAK is still unknown.



It was first discovered in~\cite{gamarnik2009sequential} an intrinsic connection between efficient computation of capacities of two-dimensional codes and the property of being \emph{strong spatial mixing} (SSM), a notion originated from the phase transition of correlation decay in statistical physics. Being strong spatial mixing means the correlation between any bits that are far away from each other decays rapidly, and hence the capacity can be efficiently estimated from local structure. 


In a seminar work~\cite{weitz},  the strong spatial mixing is introduced and is proved for independent sets of graphs along self-avoiding walks, another essential object in statistical physics~\cite{madras2013self}. Specifically, for independent sets, the maximum degree 5 is a phase transition threshold such that strong spatial mixing holds for all graphs with maximum degree at most 5, but there are graphs of any maximum degree greater than 5 without spatial mixing.
A direct consequence is an efficient algorithm for approximately computing hard-square (HS) entropy, because the hard-square constraint can be interpreted as independent sets of two-dimensional grid, whose degree is less than 5.
For the more complicated constraints of HH, RWIM, and NAK, which correspond to independent sets of graphs of degree 6 or 8, we need stronger tools than the generic ones used in~\cite{weitz} to verify the existence of strong spatial mixing and efficient algorithm for computing the capacity, or show evidence saying that they may not exist.

\subsection{Contributions}
Previously it is known that strong spatial mixing holds for the hard-square constraint~\cite{weitz} and there exists a polynomial-time approximation scheme (PTAS) for computing its capacity~\cite{pavlov, marcus}.
To analyze the spatial mixing of two-dimensional codes arising from the aforementioned constraints, we apply several techniques from the state of the art of counting algorithms and statistical physics, including: self-avoiding walk tree~\cite{weitz}, sequential cavity method~\cite{gamarnik2009sequential}, branching matrix~\cite{restrepo}, the potential function proposed in~\cite{li2}, connective-constant-based strong spatial mixing~\cite{sinclair}, and the necessary condition for correlation decay in~\cite{vera}. We make the following discoveries:
\begin{enumerate}
\item Strong spatial mixing holds for the hard-hexagon (HH) and the read/write isolated memory (RWIM) constraints.
\item Consequently, there exist PTAS for computing the capacities of HH and RWIM constraints.\footnote{Although the hard-hexagon entropy is known to be exact solvable due to its special structure~\cite{baxter}, we remark that the strong spatial mixing of this important model is interesting by itself.}
\item For the non-attacking-kings (NAK) constraint, strong spatial mixing does not hold along self-avoiding walks .
\end{enumerate}
This gives the first algorithm with provable efficiency for computing the capacity of RWIM constraint and the first strong spatial mixing results for both hard-hexagon and RWIM constraints, and also shows that the NAK constraint might not enjoy sufficient spatial mixing to support efficient computation of capacity.

\subsection{Related work}

Computing the capacities for different constrained codes has been studied extensively. In~\cite{baxter}, the exact solution of hard-hexagon entropy was given. The method introduced in~\cite{calkin} connects the number of independent sets to the capacity and shows the existence and bounds for capacities of certain important two-dimensional $(d,k)$ run-length constraints. In~\cite{kato}, the $(d,k)$ run-length constraints with positive capacities are fully characterized. In~\cite{nagy}, the bounds for the capacity of three-dimensional $(0,1)$ run-length constrained channel is given.
In~\cite{weeks}, numerical bounds on the capacity of NAK was given. In \cite{cohn}, the bounds for the RWIM capacity is given.
In~\cite{sabato}, belief propagation is used to analyze the capacities of two-dimensional and three-dimensional run-length limited constraint codes.
In~\cite{gamarnik2009sequential}, sequential cavity method was used to show PTAS for computing the free energy and surface pressure for various statistical mechanics models on $\mathbb{Z}^d$, which covers hard-square entropy and matching.
In~\cite{pavlov}, a PTAS for computing the hard-square entropy is given by ergodic theoretic techniques and methods from percolation theory.
In \cite{marcus}, it is proved that if any nearest neighbor two-dimensional shift of finite types exhibits SSM, then there is a PTAS for computing the entropy.

The strong spatial mixing was introduced in~\cite{weitz} for counting algorithms. The self-avoiding walk tree was introduced in~\cite{weitz} to deal with Boolean-state pairwise constraints, and was generalized in~\cite{bayati} and~\cite{nair2007correlation} into its full-fledged power to deal with matching, and multi-state multi-wise constraints.
These techniques were improved in a series of works~\cite{restrepo,li2,sinclair,sinclair2,vera}.

\newcommand{\cons}{{\mathit{Cons}}}
\newcommand{\lattice}{\mathbb{L}}
\newcommand{\HS}{{\mathrm{HS}}}
\newcommand{\HH}{{\mathrm{HH}}}
\newcommand{\RWIM}{{\mathrm{RWIM}}}
\newcommand{\NAK}{{\mathrm{NAK}}}
\newcommand{\IS}{{\#\mathrm{IS}}}

\section{Preliminaries}
\subsection{Two-dimensional codes from hard constraints}
A two-dimensional binary codeword is an $m\times n$ matrix of Boolean (0 and 1) entries. Let $P$ be a $k\times\ell$ Boolean matrix, called a pattern.
A two-dimensional binary codeword $W$ is said to contain pattern $P$ if $P$ is a submatrix of $W$, respecting the relative positions. Formally, there exist $i$ and $j$ such that for any $1\le s\le k$ and $1\le t\le \ell$, it holds that $W(i+s-1,j+t-1)=P(s,t)$. 
We consider the following two-dimensional codes defined by forbidding certain patterns.
\begin{enumerate}
\item Hard square (HS) constraint: A codeword does not contain patterns
\begin{equation*}
\begin{pmatrix}
1 & 1
\end{pmatrix}
\text{ and }
\begin{pmatrix}
1 \\
1
\end{pmatrix}.
\end{equation*}
The constraint forbids any horizontally or vertically consecutive 1's.

\item Hard hexagon (HH) constraint: A codeword does not contain patterns
\begin{equation*}
\begin{pmatrix}
1 & 1
\end{pmatrix},
\begin{pmatrix}
1 \\
1
\end{pmatrix}
\text{ and }
\begin{pmatrix}
0 & 1 \\
1 & 0
\end{pmatrix}.
\end{equation*}
The constraint forbids any horizontally, vertically, or anti-diagonally consecutive 1's.

\item Read/write isolated memory (RWIM) constraint: A codeword does not contain patterns
\begin{equation*}
\begin{pmatrix}
1 & 1
\end{pmatrix},
\begin{pmatrix}
1 & 0 \\
0 & 1
\end{pmatrix}
\text{ and }
\begin{pmatrix}
0 & 1 \\
1 & 0
\end{pmatrix}.
\end{equation*}
The constraint forbids any horizontally, diagonally or anti-diagonally consecutive 1's.
The horizontal pattern corresponds to the \textit{read restriction}: no two consecutive positions may store 1's simultaneously.
The two diagonal patterns correspond to the \textit{write restriction}: no two consecutive positions in the memory can be changed during one rewriting phase.

\item Non-attacking kings (NAK) constraint: A codeword does not contain patterns
\begin{equation*}
\begin{pmatrix}
1 & 1
\end{pmatrix},
\begin{pmatrix}
1 \\
1
\end{pmatrix},
\begin{pmatrix}
0 & 1 \\
1 & 0
\end{pmatrix}
\text{ and }
\begin{pmatrix}
1 & 0 \\
0 & 1
\end{pmatrix}.
\end{equation*}
The constraint forbids any horizontally, vertically, diagonally or anti-diagonally consecutive 1's.
\end{enumerate}

Throughout the paper, we assume $\cons\in\{\HS,\HH,\RWIM,\NAK\}$ to be one of the constraints defined as above. Let $N^\cons_{m,n}$ be the number of $m\times n$ Boolean matrices satisfying constraint $\cons$.
We observe that the above two-dimensional codes can be equivalently defined as independent sets of certain lattice graphs. In fact, restricting to the two-dimensional codes defined by \concept{local} forbidden patterns (of dimension up to 2), they are the all four cases which can be described as independent sets.\footnote{Other forbidden patterns of dimension up to 2 may also define independent sets, such as the one described by $\begin{pmatrix}
0 & 1 \\
1 & 0
\end{pmatrix}$ and $\begin{pmatrix}
1 & 0 \\
0 & 1
\end{pmatrix}$, however, this case is just a union of two disjoint instances of hard square.}

Let $\mathbb{Z}$ be the integer field. Let $\lattice^\cons=(\mathbb{Z}^2,E^\cons)$ be the infinite lattice graph with vertex set $\mathbb{Z}^2$ whose edge set $E^\cons$ is defined as follows, respectively:
\begin{align*}
E^{\HS}
&= \{((i,j),(i+1,j)),((i,j),(i,j+1)): \forall i,j \in \mathbb{Z}^2\};\\
E^{\HH}
&= \{((i,j),(i+1,j)),((i,j),(i,j+1)),((i,j),(i+1,j-1)): \forall i,j \in \mathbb{Z}^2\};\\
E^{\RWIM}
&= \{((i,j),(i,j+1)),((i,j),(i+1,j+1)),((i,j),(i+1,j-1)): \forall i,j \in \mathbb{Z}^2\};\\
E^{\NAK}
&= \{((i,j),(i+1,j)),((i,j),(i,j+1)),((i,j),(i+1,j+1)),((i,j),(i+1,j-1)): \forall i,j \in \mathbb{Z}^2\}.
\end{align*}
The lattices $\lattice^\HS$ and $\lattice^\HH$ are just two-dimensional grid lattice and hexagonal lattice, respectively. The lattices $\lattice^\HH$, $\lattice^\RWIM$, and $\lattice^\NAK$ are shown in Figure~\ref{fig:lattice}.

\begin{figure*}
\centering
\begin{minipage}[b]{0.3\textwidth}
\includegraphics[width=1\textwidth]{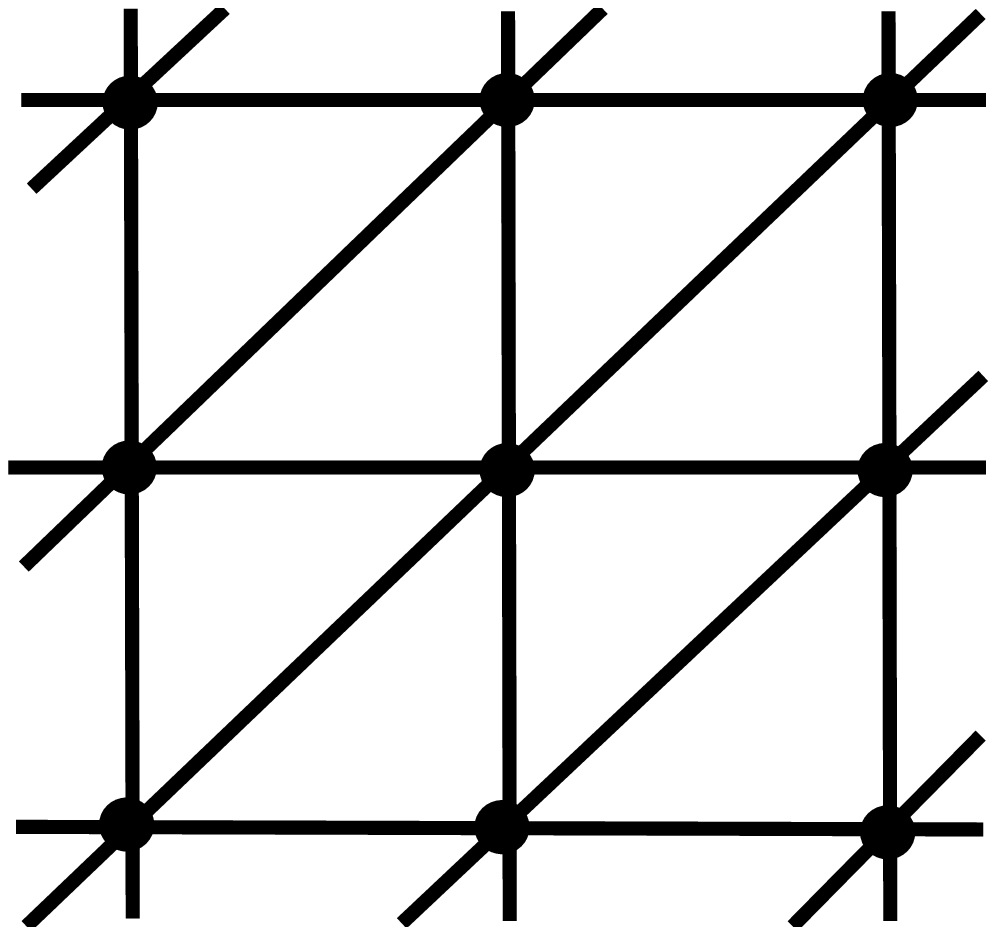}
\end{minipage}
\begin{minipage}[b]{0.3\linewidth}
\includegraphics[width=1\textwidth]{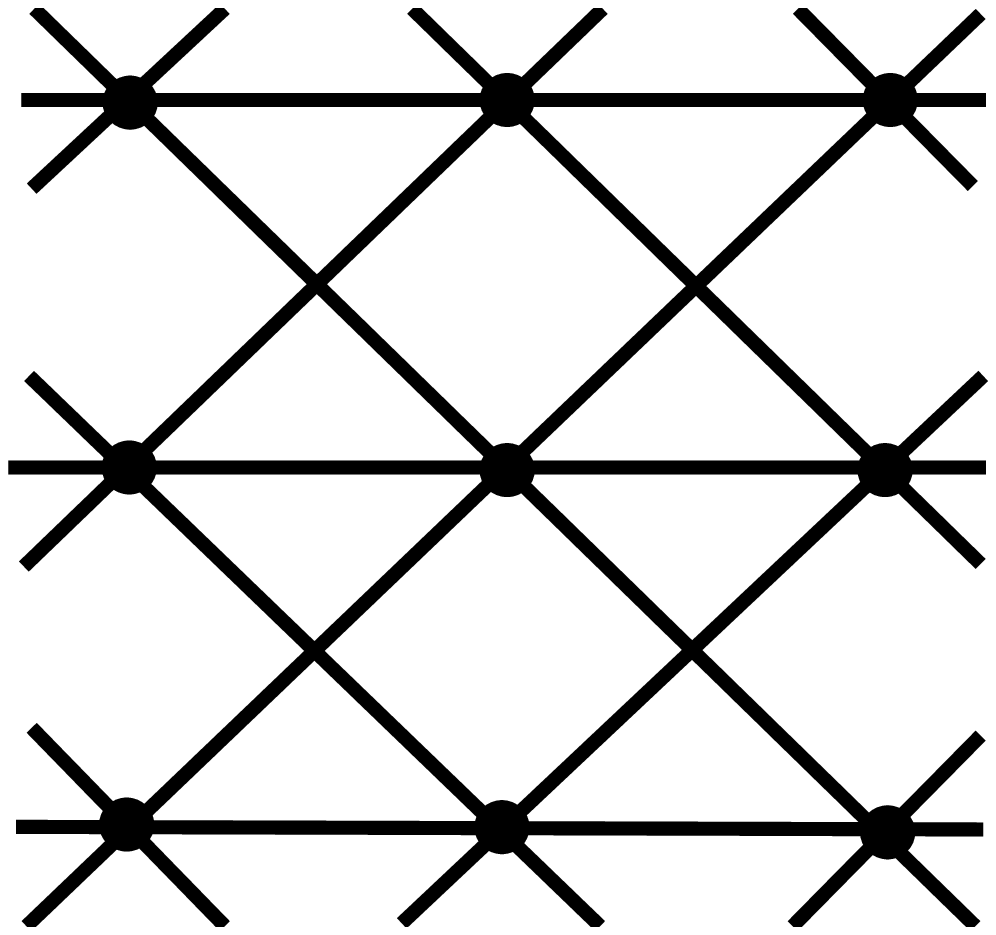}
\end{minipage}
\begin{minipage}[b]{0.3\textwidth}
\includegraphics[width=1\textwidth]{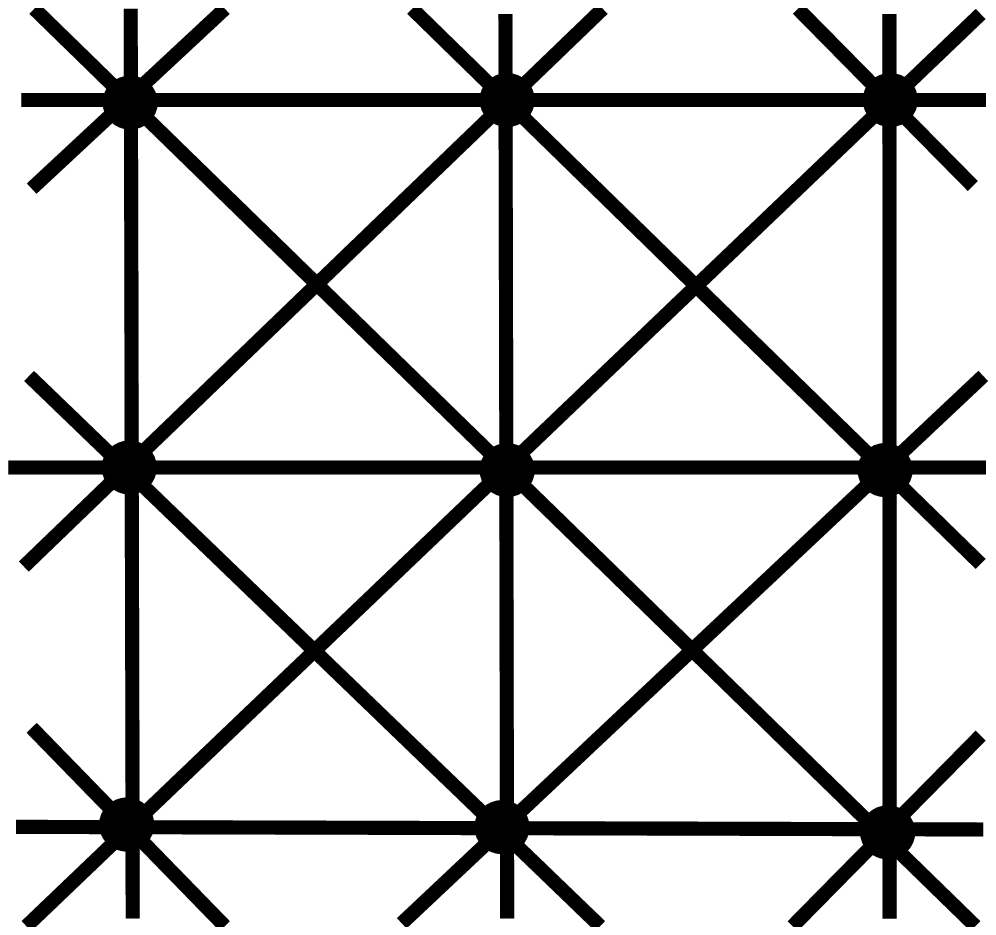}
\end{minipage}
\caption{Local structures of $\mathbb{L}^{\text{HH}}$, $\mathbb{L}^{\text{RWIM}}$ and $\mathbb{L}^{\text{NAK}}$}\label{fig:lattice}

\end{figure*}

We denote by $\lattice^\cons_{m,n}=\lattice^\cons[[m]\times[n]]$ the induced subgraph of $\lattice^\cons$ on the finite vertex set $[m]\times[n]$.
It is easy to verify that for the considered constraints $\cons\in\{\HS,\HH,\RWIM,\NAK\}$, we have
\[
N^\cons_{m,n}=\IS(\lattice^\cons_{m,n}),
\]
where $\IS(G)$ denotes the number of independent sets of graph $G$.  Moreover, each codeword of size $m\times n$ satisfying constraint $\cons$ corresponds to a unique independent set of lattice $\lattice^{\cons}_{m,n}$ such that the 1's in the codeword indicates the vertices in the independent set.

\begin{Def}
Let $\cons\in\{\HS,\HH,\RWIM,\NAK\}$. The \concept{capacity} of constraint $\cons$, denoted $C_\cons$, is defined by
\begin{align*}
C_{\cons} = \lim_{m,n \to \infty} \frac{\log_2 N^\cons_{m,n}}{m n}=
\lim_{m,n \to \infty} \frac{\log_2 \#\mathrm{IS}(\lattice^\cons_{m,n})}{m n}.
\end{align*}
\end{Def}

\subsection{Spatial Mixing}
We adopt notions and terminologies from statistical physics to describe the probability space of independent sets. Let $G(V,E)$ be a graph, where each vertex is in one of the two states $\{0,1\}$, such that state 1 is called \concept{occupied} and state 0 is called \concept{unoccupied}.
Each \concept{configuration} $\sigma\in\{0,1\}^V$ indicates a subset $I_\sigma\subseteq V$ of vertices such that a vertex is occupied means it is in $I_\sigma$. 
Given a finite graph $G(V,E)$, the uniform probability measure $\mu_G$ over all independent sets of $G$ can be defined as
\[
\forall \sigma\in\{0,1\}^V,\quad
\mu_G(\sigma)=\begin{cases}
\frac{1}{Z_G} & I_\sigma\mbox{ is an indepdent set in }G,\\
0 & \mbox{otherwise},
\end{cases}
\]
where $Z_G=\IS(G)$ is the number of independent sets of $G$. Here $Z_G$ is also called the \concept{partition function}, and $\mu_G$ is called the \concept{Gibbs measure} of independent sets.

Given $G(V,E)$, let $v\in V$ be a vertex and $\rho\in\{0,1\}^\Lambda$ an independent set of vertex set $\Lambda\subset V$.  Let $p_{G,v}^\rho$ denote the marginal probability of $v$ being unoccupied conditioned on configuration $\rho$, that is
\begin{align*}
p_{G,v}^\rho=\Pr_{\sigma\sim\mu_G}[\sigma_v=0\mid \sigma_{\Lambda}=\rho].
\end{align*}
For this conditional probability space, we say that the vertices in $\Lambda$ are \concept{fixed} to be $\rho$.

For infinite graph $G$, the uniform probability measure $\mu_G$ cannot be defined since uniform distribution cannot be defined on countably infinite set. So we deal with finite subgraphs. 

\concept{Spatial mixing}, also called the \concept{correlation decay}, is the property that the influence of an arbitrarily fixed boundary on the marginal distribution on a vertex decreases  exponentially as the distance between them grows. It is one of the essential concepts in Statistical Physics.

\begin{Def}[Weak Spatial Mixing] The independent sets of an infinite graph $G(V,E)$ exhibits \concept{weak spatial mixing (WSM)} if there exist constants $\beta,\gamma>0$, such that for every finite vertex set $U \subset V$, every $v \in U$, and any two independent sets $\sigma$, $\tau$ of the vertex boundary $\partial U = \{w\not\in U|uw\in E,u\in U\}$, it holds that
\begin{equation*}
\left|p_{G,v}^{\sigma} - p_{G,v}^{\tau}\right| \le \beta\cdot\exp(-\gamma\cdot\text{dist}(v,\partial U)),
\end{equation*}
where $\text{dist}(v,\partial U)$ is the shortest distance between $v$ and any vertex in $\partial U$.
\end{Def}

In order for algorithmic applications, we need a stronger version of spatial mixing, called the \concept{strong spatial mixing}, which is introduced in~\cite{weitz}.

\begin{Def}[Strong Spatial Mixing] The independent sets of an infinite graph $G(V,E)$ exhibits \concept{strong spatial mixing (SSM)} if there exist constants $\beta,\gamma>0$, such that for every finite vertex set $U \subset V$, every $v \in U$, and any two independent sets $\sigma$, $\tau$ of the vertex boundary $\partial U = \{w\not\in U|uw\in E,u\in U\}$, it holds that
\begin{equation*}
\left|p_{G,v}^{\sigma} - p_{G,v}^{\tau}\right| \le \beta\cdot\exp(-\gamma\cdot\text{dist}(v,\Delta)),
\end{equation*}
where $\Delta \subseteq \partial U$ is the set of vertices on which $\sigma$ and $\tau$ differ and $\text{dist}(v,\Delta)$ is the shortest distance between $v$ and any vertex in $\Delta$.
\end{Def}

The difference between WSM and SSM is that SSM requires that the correlation decay still holds even with the configuration of a subset $\partial U\setminus\Delta$ of nearby vertices to $v$ to be arbitrarily fixed.
It is easy to see that SSM implies WSM. Moreover, we have the following easy but useful proposition for independent sets.

\begin{Prop}\label{prop:1}
If an infinite tree $T$ exhibits SSM, then all subtrees of $T$ exhibit SSM.
\end{Prop}

The proposition is implied by the simple observation that fixing vertices to be unoccupied effectively prunes the tree to an arbitrary subtree, while the SSM still holds.

\subsection{Self-avoiding walk tree}\label{section-SAW}
\newcommand{\saw}{\mathrm{SAW}}

On trees, there is an easy recursion for marginal probabilities.
Let $T$ be a tree rooted by $v$, and $v_1,v_2,\ldots,v_d$ the children of $v$. For each $1\le i\le d$, let $T_i$ denote the subtree of $T$ rooted by $v_i$. For any independent set $\sigma$ of a subset $\Lambda$ of vertices in $T$, consider the ratio $R_{T}^\sigma = (1-p^\sigma_{T,v})/p^\sigma_{T,v}$ between marginal probabilities of being occupied and unoccupied at vertex $v$. The following recursion is well known:
\begin{equation}
R_{T}^{\sigma} = \prod_{i=1}^d \frac{1}{1 + R_{T_i}^{\sigma_i}}, \label{eq-ratio-recursion}
\end{equation}
where $\sigma_i$ is the restriction of $\sigma$ on $T_i$.

The \concept{self-avoiding walk (SAW) tree} is introduced in \cite{weitz} to transform a graph into a tree while preserving the marginal probability.
Let $G(V,E)$ be a graph, finite or infinite. For each vertex $u\in V$, we fix an arbitrary order $>_u$ for the neighbors of $u$. Let $v\in V$ be an arbitrary vertex.
A tree $T$ rooted by $v$ can be naturally constructed from all self-avoiding walks $v =v_0\to v_1 \to \ldots \to v_\ell $ starting from vertex $v\in V$, 
after which for any walk $v\to\ldots \to w \to v_k \to\ldots \to v_\ell$ such that $\{v_\ell, w\}\in E$ and $v_\ell >_w v_k$, we delete the corresponding node and the subtree from $T$. The resulting tree is denoted as $T=T_{\saw}(G,v)$. This construction identifies each vertex in $T$ (many-to-one) to a vertex in $G$. Thus for any independent set $\sigma$ of $\Lambda\subset V$, we have a corresponding configuration in $T$, which is still denoted as $\sigma$ by abusing the notation.

\begin{Thm}[Weitz~\cite{weitz}]\label{thm-saw}
For any finite graph $G(V,E)$, $v \in V$, $\Lambda \subset V$ and any independent set $\sigma$ of $\Lambda$, it holds that 
\begin{equation*}
p_{G,v}^{\sigma} = p_{T,v}^{\sigma},
\end{equation*}
where $T = T_{\saw}(G,v)$.
\end{Thm}
In~\cite{bayati} and~\cite{nair2007correlation}, the self-avoiding walk tree is generalized to deal with general constraints.


\section{Computing the capacity by SSM}
We apply the sequential cavity method of Gamarnik and Katz~\cite{gamarnik2009sequential}, which gives efficient approximation algorithm for computing the entropy of lattice models in statistical mechanics using the strong spatial mixing. In~\cite{marcus}, this was used for approximately computing the entropy of the 2-dimensional Markov random fields exhibiting strong spatial mixing.
These approximation algorithms, along with the one given in~\cite{pavlov} for the hard-square entropy, are all in the framework of \concept{polynomial-time approximation scheme (PTAS)}, which is defined as follows. 


\begin{Def}
We say that there exists a \concept{polynomial-time approximation scheme (PTAS)} for computing a real number $C\in[0,1]$ if for any $\epsilon>0$ a number $\hat{C}$ can be returned in time $\mathrm{poly}(\frac{1}{\epsilon})$ such that $|\hat{C}-C|\le\epsilon$.
\end{Def}




Here we restate the proof in~\cite{gamarnik2009sequential} for the implication from SSM to the existence of PTAS in our context for self-containness.

\begin{Thm}\label{thm:ssm-ptas}
For $\cons\in\{\HS,\HH,\RWIM,\NAK\}$, if the independent sets of $\mathbb{L}^{Cons}$ exhibit SSM, then we have a PTAS for computing $C_\cons$.
\end{Thm}

We define some notations. Consider the vertex set $[2t+1]\times [2t+1]$ of the finite lattice $\lattice_{2t+1,2t+1}^\cons$. Let $Q= \{(k,\ell): k \in [t-1], \ell \in [2t+1]\} \cup \{(k,\ell): k = t, \ell \in [t-1]\}$ denote the first $t-1$ rows and the first $t-1$ entries of the $t$-th row in $[2t+1]\times [2t+1]$, and let $\textbf{0}_{Q}\in\{0,1\}^{Q}$ be the configuration fixing all vertices in $Q$ to be unoccupied.
Consider uniformly distributed independent set $\sigma$ of $\lattice_{2t+1,2t+1}^\cons$. Let $p_t = \Pr[\sigma_{t,t}=0|\textbf{0}_{Q}]$ be the marginal probability of central point $(t,t)$ being unoccupied conditioned on that all vertices in $Q$ being fixed to be unoccupied.

\begin{Lem}\label{lemma:ssm-ptas}
If the independent sets of $\lattice^{\cons}$ exhibit SSM, then for any constant $0<\epsilon<1$, there is a $t=O(\log\frac{1}{\epsilon})$ such that
\[
\left|\log_2\frac{1}{p_t}-C_\cons\right|\le\epsilon.
\]
\end{Lem}

\begin{proof}

Let $n\gg t$ be sufficiently large. We use the finite lattice $\lattice_{n,n}^\cons$ on vertex set $[n]\times[n]$ to connect the marginal probability $p_t$ in a constant size ($t=O(\log\frac{1}{\epsilon})$) instance to the capacity $C_\cons$ defined on the infinite lattice $\lattice^\cons$.

For each $i,j\in[n]$, let $S_{i,j}=\{(k,\ell)\in[n]\times [n]:|k - i| \le t, |l - j| \le t\}$ denote the $(2t+1) \times (2t+1)$ square centered at $(i,j)$ (truncated if it goes beyond the boundary of $[n]\times[n]$. Let $S'$ denote the set of those vertices $(i,j)\in[n]\times[n]$ whose $S_{i,j}$ are not truncated by the boundary of $\lattice_{n,n}^\cons$, that is, $S' = \{(i,j)\in[n]\times[n]: t < i,j \le n-t\}$.

For each $i,j\in[n]$, let $Q_{i,j}= \{(k,\ell)\in[n]\times [n]: k \in [i-1], \ell \in [n]\} \cup \{(k,\ell)\in[n]\times [n]: k = i, \ell \in [j-1]\}$ denote the first $i-1$ rows and the first $j-1$ entries of the $i$-th row in $[n]\times [n]$, and let $\textbf{0}_{i,j}\in\{0,1\}^{Q_{i,j}}$ be the configuration fixing all vertices in $Q_{i,j}$ to be unoccupied.
Consider uniformly distributed independent set $\sigma$ of $\lattice_{n,n}^\cons$. For each $i,j\in[n]$, let $p_{i,j} = \Pr[\sigma_{i,j}=0|\textbf{0}_{i,j}]$ be the marginal probability of vertex $(i,j)$ being unoccupied conditioned on that all vertices in $Q_{i,j}$ being fixed to unoccupied.

Let $\mathbf{0}$ be the configuration fixing all vertices in $[n]\times [n]$ to be unoccupied. For a uniformly distributed independent set $\sigma$ of $\lattice^\cons_{n,n}$, we have $\Pr[\sigma=\mathbf{0}]=1/N^\cons_{n,n}$ where $N^\cons_{n,n}$ denotes the number of independent sets of $\lattice^\cons_{n,n}$. Moreover, it holds that
\begin{align*}
\Pr[\sigma=\mathbf{0}]
&=\prod_{i,j\in[n]}\Pr[\sigma_{i,j}=0\mid \mathbf{0}_{i,j}]
=\prod_{i,j\in[n]}p_{i,j}.
\end{align*}
Therefore, we have
\begin{align*}
\frac{\log_2 N^\cons_{n,n}}{n^2}=\frac{1}{n^2}\sum_{i,j\in[n]}\log_2\frac{1}{p_{i,j}}.
\end{align*}
Let $p_{i,j}'=\Pr[\sigma_{i,j}=0|\textbf{0}_{i,j}\wedge \textbf{0}_{B}]$ where $\textbf{0}_{B}$ is the configuration fixing all boundary vertices of $S_{i,j}$ in $\lattice_{n,n}^\cons$ to be unoccupied.
Note that for $(i,j) \in S'$ and any considered constraint $\cons$, the distance in $\lattice^{cons}$ is at least the grid distance distorted by a constant factor, thus the shortest distance between vertex $(i,j)\in S'$ and the boundary is $\Omega(t)$.
Suppose that the independent sets of $\lattice^\cons_{n,n}$ exhibits SSM. Then for some $t=O(\log\frac{1}{\epsilon})$, we have $|p_{i,j}-p_{i,j}'|\le \epsilon$ for all $(i,j)\in S'$.

Furthermore, for those $(i,j)\in S'$, it is easy to see that $p'_{i,j}=p_t$. And for $(i,j)\in \lattice_{n,n}^\cons \setminus S'$, it holds that $p_{i,j}\in[\frac{1}{2},1]$. Therefore, we have
\begin{align*}
\frac{(n-2t)^2 \log_2\frac{1}{p_t+\epsilon}}{n^2}
\le
\frac{1}{n^2}\sum_{i,j\in[n]}\log_2\frac{1}{p_{i,j}}\le
\frac{(n-2t)^2 \log_2\frac{1}{p_t-\epsilon}+(4nt-t^2)}{n^2}.
\end{align*}
when $n\rightarrow\infty$, we have $C_\cons=\frac{\log_2 N^\cons_{n,n}}{n^2}=\frac{1}{n^2}\sum_{i,j\in[n]}\log_2\frac{1}{p_{i,j}}$ and
\[
\log_2\frac{1}{p_t+\epsilon}
\le
C_\cons
\le
\log_2\frac{1}{p_t-\epsilon}.
\]
Therefore,
\begin{align*}
\left|\log_2\frac{1}{p_t}-C_\cons\right|
\le\log_2\frac{1}{p_t-\epsilon}-\log_2\frac{1}{p_t+\epsilon}
=\log_2\left(1+\frac{2\epsilon}{p_t-\epsilon}\right)
\le\log_2\left(1+\frac{2\epsilon}{1/2-\epsilon}\right)
=O(\epsilon).
\end{align*}
\end{proof}

The exact value of $p_t$ can be relatively efficiently computed because the graph on which $p_t$ is defined has bounded treewidth. Precisely, for all considered constraints $\cons$, the treewidth of the finite graph $\lattice_{2t+1,2t+1}^\cons$ is $O(t)$. And the independent set is covered by the framework considered in~\cite{yin}. The value of $p_t$ can thus be computed exactly by the dynamic programming algorithm introduced in~\cite{yin} with time complexity $2^{O(t)} \cdot poly(t) = \mathrm{poly}(\frac{1}{\epsilon})$.
Combined with Lemma~\ref{lemma:ssm-ptas}, Theorem~\ref{thm:ssm-ptas} is proved.

\section{SSM of hard-hexagon and RWIM}
It is well known that the independent sets of two-dimensional grid $\lattice^\HS$ exhibits SSM~\cite{weitz}.
We now prove the following theorem.
\begin{Thm}\label{thm-ssm-HH-RWIM}
The independent sets of $\lattice^\HH$ and $\lattice^\RWIM$ exhibits SSM.
\end{Thm}
It is well known that $C_\HH$ is exact solvable~\cite{baxter}. Applying Theorem~\ref{thm:ssm-ptas}, we have the following new algorithmic result for $C_\RWIM$. 
\begin{Cor}
There exists PTAS for computing $C_\HH$ and $C_\RWIM$.
\end{Cor}

\subsection{Branching matrix}
The SSM is proved on a supertree of the self-avoiding walk tree for the respective lattice. This supertree is a multi-type Galton-Watson tree generated by a branching matrix whose definition is introduced in~\cite{restrepo}.

\begin{Def}[Restrepo \emph{et al.}~\cite{restrepo}]
A branching matrix $M$ is an $m \times m$ matrix of nonnegative integral entries. Each branching matrix represents a rooted tree $T_M$ generated by the following rules:
\begin{itemize}
\item each node of $T_M$ is in one of the $m$ types $\{1,\ldots,m\}$, with the root being type 1;
\item every type-$i$ node has exactly $M_{ij}$ many children of type $j$.
\end{itemize}
\end{Def}

Note that if $M$ is irreducible, then $T_M$ is an infinite tree. The maximum arity of the tree is $d = \max_{1 \le i \le m} \sum_{j = 1}^m M_{ij}$. The following lemma gives a relation between SSM and  the maximum eigenvalue of branching matrix.

The SSM is proved by the following main lemma, which can be implied by Theorem~1.3 of~\cite{sinclair} through the notion of \emph{connective constant}. Here we restate the proof to the lemma without using connective constant. The following lemma relates the SSM to the maximum eigenvalue of branching matrix.
\begin{Lem}[implicit in~\cite{sinclair}]\label{lemma-bm-ssm}
Let $G(V,E)$ be an infinite graph with maximum degree $d+1$. Let $\gamma=\inf_{x \in [0,+\infty)}\frac{[1+(1+x)^d](1+x)}{dx}$.
If for every $v\in V$, there exists a branching matrix $M$ satisfying:
\begin{enumerate}
\item the tree $T_M$ generated by $M$ is a supertree of $T_{\saw}(G,v)$, and
\item the largest eigenvalue $\lambda^*$ of $M$ is less than $\gamma$,
\end{enumerate}
then the independent sets of $G$ exhibits SSM.
\end{Lem}

\subsection{Potential analysis of correlation decay}

By Theorem~\ref{thm-saw}, the SSM on graph $G$ can be implied by the SSM on its SAW tree, which can be implied by the SSM on any supertree of the SAW tree, according to Proposition~\ref{prop:1}. We then verify the SSM for the supertree $T_M$ of the SAW tree generated by a branching matrix $M$ whose eigenvalues satisfy the condition of Lemma~\ref{lemma-bm-ssm}. This can be done by showing that the system~\eqref{eq-ratio-recursion} on the tree $T_M$ converges at an exponential rate while the boundary conditions are arbitrarily fixed. Two key ideas for the analysis are:
\begin{itemize}
\item Instead of analyzing the convergence of the ratios $R = \frac{1-p}{p}$ of the marginal probabilities as straightforwardly in the recursion~\eqref{eq-ratio-recursion}, we analyze the convergence of the ``potentials" $\phi = \sinh^{-1}(\sqrt{R})$. This potential function was introduced in~\cite{li2} and later used in~\cite{sinclair}, seeming to capture the very nature of hard (independent set) constraint.
\item We use the $l_2^2$-norm (sum of squares) to measure the errors of potentials for vertices of different types. The same scheme was proposed in~\cite{sinclair}.
\end{itemize}

Without loss of generality, we consider $m\times m$ branching matrix $M$ of Boolean (0 and 1) entries. For branching matrices with entries greater than 1, we can refine each of the involved types to a number of types so that the resulting branching matrix has Boolean entries, generates the same tree as before, and has the same largest eigenvalue.
Let $T_M$ be the infinite tree generated by $M$.
A set $S$ of vertices in $T_M$ is a cutset if any infinite path in $T_M$ intersects with $S$.
Let $\sigma$ denote an independent set of cutset $S$. For any vertex $v$ in $T_M$, let $T_v$ be the subtree rooted by $v$ and $x_v$ the ratio between probabilities of $v$ being occupied and unoccupied in $T_v$ conditioned on vertices in $S$ being fixed as $\sigma$. Suppose $v$ is of type $i$. By~\eqref{eq-ratio-recursion}, we have the following recursion:
\begin{equation}
x_v= \label{branching-recursion}
\begin{cases}
\frac{1}{\prod_{j = 1}^m (1 + x_j)^{M_{ij}}} & \text{ if }v \notin S, \\
0 & \text{ if } v \in S \text{ and }\sigma_v = 0, \\
\infty & \text{ if } v \in S \text{ and }\sigma_v = 1.
\end{cases}
\end{equation}
where $x_j$ is the corresponding ratio of probabilities at the child of type $j$ in its subtree. Note that $M_{ij}\in\{0,1\}$ and the function depends only on those $x_j$ with $M_{ij}\neq 0$, thus the recursion is well-defined. Let $\eta$ be another independent set of cutset $S$. We can define a new sequence of quantities $x_v'$ in the same way.

Consider any finite graph $G$ and a vertex $v$, and any two independent sets $\rho,\eta$ of vertex subset $\Lambda$ where $\rho,\eta$ disagree only at vertices at least $\ell$ far away from $v$. Let $T=T_\saw(G,v)$ and $T_M$ a supertree of $T$. Then by the same argument as in Proposition~\ref{prop:1}, there must exist two independent sets $\sigma,\tau$ of a cutset $S$ of $T_M$, disagreeing only at vertices at  least $\ell$ far away from the root $v$ of $T_M$, such that $p_{T,v}^\rho=p_{T_M,v}^\sigma$ and $p_{T,v}^\eta=p_{T_M,v}^\tau$. Therefore, by Theorem~\ref{thm-saw} and definitions of the quantities $x_v$ and $x_v'$ as above, we have
\begin{align*}
\left|p_{G,v}^\rho-p_{G,v}^\eta\right|
=
\left|p_{T,v}^\rho-p_{T,v}^\eta\right|
=
\left|p_{T_M,v}^\sigma-p_{T_M,v}^\tau\right|
=
\left|\frac{1}{1+x_v}-\frac{1}{1+x_v'}\right|
\le
|x_v-x_v'|.
\end{align*}
To prove the SSM, it is then sufficient to show it always holds that $|x_v-x_v'|\le\exp(-\Omega(\ell))$ for any independent sets $\sigma,\tau$ in cutset $S$ of $T_M$ which disagree only at vertices $\ell$ far away from the root $v$.

\begin{proof}[Proof of Lemma~\ref{lemma-bm-ssm}]
For each $i\in\{1,2,\ldots,m\}$, we define
\[
f_i(\boldsymbol{x})=\frac{1}{\prod_{j = 1}^m (1 + x_j)^{M_{ij}}},
\]
where $\boldsymbol{x}=(x_1,x_2,\ldots,x_m)$. By~\eqref{branching-recursion}, if the current root $v$ is unfixed and is of type $i$, then we have $x_v=f_i(\boldsymbol{x})$ and $x_v'=f_i(\boldsymbol{x}')$. 
We define the potentials as $y_v = \phi(x_v)$ and 
$y_j = \phi(x_j)$
for all $j=1,2,\ldots, m$, where the potential function $\phi(\cdot)$ is given by
\[
\phi = \sinh^{-1}(\sqrt{x}).
\]
This is the potential function used in~\cite{li2,sinclair,liu2013fptas}. Supposed $v$ is of type $i$, the mapping from $\boldsymbol{y}$ to $y_v$ can be deduced as
\begin{align*}
y_v=g_i(\boldsymbol{y}) =
\phi(f_i(\phi^{-1}(y_1),\phi^{-1}(y_2),\ldots,\phi^{-1}(y_m))).
\end{align*}
Since $\phi$ is strictly monotone, the inverse $\phi^{-1}(y)$ is well-defined.

Suppose $y_v=\phi(x_v)$ and $y_v'=\phi(x_v')$ are the respective potentials at every vertex $v$ defined by the boundary conditions $\sigma,\tau$ on cutset $S$.
For $v\not\in S$, by the mean value theorem, there exist $\tilde{y}_j\in[y_j,y_j']$, and accordingly $\tilde{x}_j=\phi^{-1}(\tilde{y}_j)$, $\forall j=1,2,\ldots,m$, such that
\begin{align*}
|y_v-y_v'|
&=
|g_i(\boldsymbol{y}) - g_i(\boldsymbol{y}')|
\le
\sum_{j=1}^m\left|\frac{\partial g_i}{\partial \tilde{y}_j}\right||y_j-y_j'|
\le
\sqrt{\frac{f_i(\tilde{\boldsymbol{x}})}{1+f_i(\tilde{\boldsymbol{x}})}}\sum_{j=1}^m\sqrt{\frac{M_{ij}\tilde{x}_j}{1+\tilde{x}_j}}\sqrt{M_{ij}}|y_j-y_j'|.
\end{align*}

Due to Cauchy-Schwarz, we have
\begin{align*}
|y_v - y_v'|^2
&\le
\left(\frac{f_i(\tilde{\boldsymbol{x}})}{1+f_i(\tilde{\boldsymbol{x}})}
\sum_{j=1}^m\frac{M_{ij}\tilde{x}_j}{1+\tilde{x}_j}\right)
\sum_{j=1}^mM_{ij}|y_j-y_j'|^2.
\end{align*}
Let $d_i=\sum_{j=1}^mM_{ij}$. Note that $d_i\le d$ where $d+1$ is the maximum degree of the original graph $G$. Recall that $\gamma=\inf_{x \in [0,+\infty)}\frac{[1+(1+x)^d](1+x)}{dx}$. It can be verified that  by such definition $\gamma$ is nondecreasing in $d$.
Let $\bar{x}=\left( \prod_{j=1}^m(1+\tilde{x}_j)^{M_{ij}} \right)^{1/d_i}-1$.
We have $\frac{f_i(\tilde{\boldsymbol{x}})}{1+f_i(\tilde{\boldsymbol{x}})}=\frac{1}{1+(1+\bar{x})^{d_i}}$, and by Jensen's inequality it can be verified that $\sum_{j=1}^m\frac{M_{ij}\tilde{x}_j}{1+\tilde{x}_j}\le \frac{d_i\bar{x}}{1+\bar{x}}$. Therefore, it holds that
\begin{align}
|y_v - y_v'|^2
&\le
\frac{d_i\bar{x}}{(1+(1+\bar{x})^{d_i})(1+\bar{x})}\sum_{j=1}^mM_{ij}|y_j-y_j'|^2
\le \frac{1}{\gamma}\sum_{j=1}^mM_{ij}|y_j-y_j'|^2.
\end{align}

Fix a cutset $S$ and a $\Delta\subseteq S$ such that the shortest distance from any vertex in $\Delta$ to the root of the tree is $\ell$. For each $1\le t\le \ell$ and $i\in\{1,2,\ldots,m\}$, let $\epsilon^{(t)}$ be an $m$-vector such that $\epsilon^{(t)}_i$ is the maximum potential difference square $|y_v-y_v'|^2$ for any vertex $v$ of type $i$ at depth $\ell-t$ (assuming the root has depth 0) defined by any two boundary conditions $\sigma,\tau$ on $S$ disagreeing only on $\Delta$. Suppose that $\epsilon^{(t)}_i=|y_v-y_v'|^2$ for a vertex of type $i$ at level $\ell -t$ for a particular pair of boundary conditions $\sigma,\tau$ on $S$. We have
\[
\epsilon^{(t)}_i
=|y_v - y_v'|^2
\le \frac{1}{\gamma}\sum_{j=1}^mM_{ij}|y_j-y_j'|^2
\le \frac{1}{\gamma}\sum_{j=1}^mM_{ij}\epsilon^{(t-1)}_{j}.
\]
Therefore, for $1\le t\le \ell$, we have the following entry-wise inequality between vectors  $\epsilon^{(t)}$ and $\epsilon^{(t-1)}$:
\[
\epsilon^{(t)}
\le
\frac{1}{\gamma}M\epsilon^{(t-1)}.
\]
It can be verified that every entry of $\epsilon^{(1)}$ is bounded by a sufficiently large constant $C$ since after one step of recursion, both $y_v$ and $y_v'$ must be bounded. Then $\epsilon^{(t)}\le \frac{C}{\gamma^t}M^t\boldsymbol{1}$. Thus if the largest eigenvalue of $M$ is less than $\gamma$,  for every $i\in\{1,2,\ldots,m\}$ it holds that $\epsilon^{(t)}_i=\exp(-\Omega(t))$ for $1<t\le \ell$. In particular for the root $|y_v-y_v'|\le \epsilon^{(\ell)}_i=\exp(-\Omega(\ell))$. Translating this back to the error between $x_v$ and $x_v'$, by the mean value theorem, there exists a $z\in[x_v,x_v']$ such that
\[
|x_v-x_v'|=|\phi^{-1}(y_v)-\phi^{-1}(y_v')| = \frac{1}{\phi'(z)}|y_v-y_v'|= \exp(-\Omega(\ell)),
\]
since the value of $\frac{1}{\phi'(z)}=2\sqrt{z(1+z)}$ is bounded. By the discussion in the beginning of this section, this proves the lemma.
\end{proof}

\subsection{Supertree construction}\label{section-supertree-construction}
For both $\mathbb{L}^{\text{HH}}$ and $\mathbb{L}^{\text{RWIM}}$ the maximum degree is 6, thus we can apply Lemma~\ref{lemma-bm-ssm} with $d=5$, which gives us $\gamma> 4.047$.
Since both $\mathbb{L}^{\text{HH}}$ and $\mathbb{L}^{\text{RWIM}}$ are symmetric for every vertex $v$, for each $\cons\in\{\HH,\RWIM\}$ we only need to construct a branching matrix $M$ satisfying:
\begin{itemize}
\item the largest eigenvalue $\lambda^*(M) < 4.047$;
\item the infinite tree $T_M$ generated by $M$ is a supertree of $T_\saw(\lattice^\cons,v)$ for an arbitrary vertex $v$.
\end{itemize}
By Lemma~\ref{lemma-bm-ssm}, this is sufficient to imply the SSM of independent sets of $\lattice^{\cons}$.


A self-avoiding walk tree $T_\saw(\lattice^\cons,v)$ contains only those walks starting from $v$ avoiding cycles. We relax this constraint and consider a tree $T^\cons_l$ containing all walks in $\lattice^\cons$ starting from $v$ avoiding cycles of length no more than a given constant $l$. Clearly $T^\cons_l$ is a supertree of $T_\saw(\lattice^\cons,v)$. Such $T^\cons_l$ can be generated by a branching matrix $M_l^\cons$ described as follows.




\begin{figure*}
\centering
\begin{minipage}[b]{0.33\textwidth}
\includegraphics[width=1\textwidth]{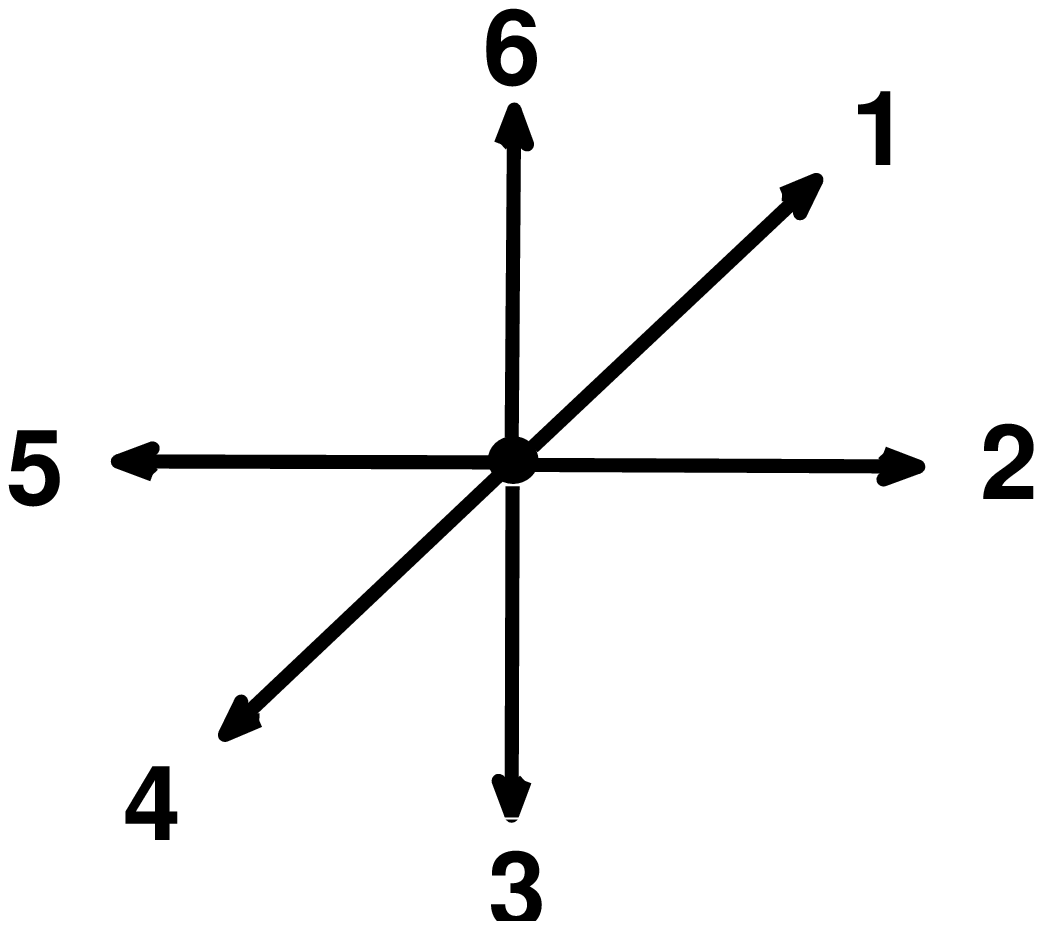}
\end{minipage}
\begin{minipage}[b]{0.33\linewidth}
\includegraphics[width=1\textwidth]{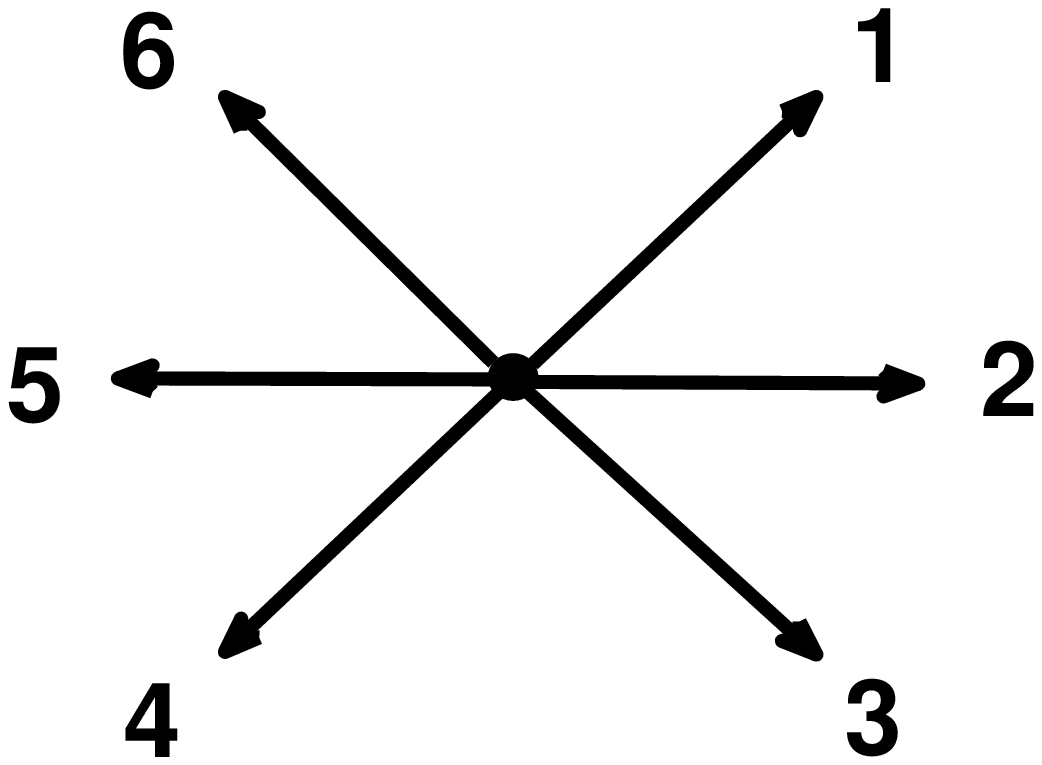}
\end{minipage}
\caption{Order of neighbors for every vertex in $\lattice^\HH$ and $\lattice^\RWIM$}\label{fig:2}
\end{figure*}


Let 
$\mathcal{C}_l$ denote the set of cycles of length no more than $l$ in $\lattice^\cons$. 
For a vertex $v$ in $T^\cons_l$, supposed $p$ to be the walk in $\lattice^\cons$ corresponding to $v$, let $p_k$ denote the last $k$ steps of $p$. If the length of $p$ is less than $k$, then $p_k=p$. The type of $v$ is identified as $p_k$ where $k$ is determined as follows.

\begin{Def}
Let $v$ be a vertex in $T^\cons_l$, whose corresponding walk in $\lattice^\cons$ is $p$. The type of $v$, denoted by $\tau(v)$, is defined as
\begin{equation*}
\tau(v) = p_k, \text{ where } k = \max\left\{i: \exists c \in \mathcal{C}_l, p_i \subseteq c \right\} ,
\end{equation*}
where $p_i \subseteq c$ means that the $i$-step walk $p_i$ equals a sequence of consecutive steps in cycle $c$.
\end{Def}


With this definition of types, we have a branching matrix $M_l^\cons$ generating $T^\cons_l$. We further refine this construction to have supertree $\widetilde{T}_l^\cons$ generated by branching matrix $\widetilde{M}_l^\cons$ which more precisely approximate the $T_\saw(\lattice^\cons,v)$. We fix an order of neighbors for every vertex in $\lattice^\cons$ as demonstrated in Figure~\ref{fig:2}. Supposed that a cycle of length $k+1\le l$ is $w \to v_1 \to \ldots \to v_k \to w$ and $v_1 <_w v_k$, by the definition of $T_\saw$ in Section \ref{section-SAW}, the subtree rooted at $v_k$ should be deleted. We call this operation the ``Effect of Order". And let $\widetilde{T}^{\cons}_l$ denote the tree resulting from deleting all such subtrees from $T^\cons_l$. Deleting the types (columns and rows) in the branching matrix $M_l^\cons$ violating this additional rule, we have the branching matrix $\widetilde{M}_l^\cons$ generating $\widetilde{T}^\cons_l$.
It is still obvious that $\widetilde{T}^{\cons}_l$ is a supertree of $T_\saw(\lattice^\cons,v)$ because it contains all self-avoiding walks in the latter.

\subsection{SSM of hard-hexagon and RWIM}
We generate the branching matrices $M_l^\cons$ and $\widetilde{M}_l^\cons$ by the rules defined as above. The data for these branching matrices is available in our online appendix~\cite{bmdata}.
The largest eigenvalues $\lambda^*$ of branching matrices $M^{\text{HH}}_l$ and $\widetilde{M}^{\HH}_l$ for $l=4,6,8$ are shown in Table~\ref{table:1} and~\ref{table:2} respectively.
Even when $l = 4$, the largest eigenvalue of $\widetilde{M}^{\HH}_l$ is less than 4.047. By Lemma \ref{lemma-bm-ssm}, the independent sets of $\lattice^\HH$ exhibit SSM.

\begin{table}[H]
  \centering
\begin{tabular}{|c|c|c|c|}
  \hline
  Max length of Avoiding-cycles & Effect of Order & Number of Types & $\lambda^*$ \\
  \hline
  4 & No & 55 & 4.5064  \\
  \hline
  6 & No & 493 & 4.3864  \\
  \hline
  8 & No & 5479 & 4.3282  \\
  \hline
\end{tabular}
\caption{The largest eigenvalues of $M^{\HH}_l$}
\label{table:1}
\end{table}

\begin{table}[H]
  \centering
\begin{tabular}{|c|c|c|c|}
  \hline
  Max length of Avoiding-cycles & Effect of Order & Number of Types & $\lambda^*$ \\
  \hline
  4 & Yes & 35 & 3.6857 \\
  \hline
  6 & Yes & 282 & 3.5872 \\
  \hline
  8 & Yes & 2858 & 3.5439 \\
  \hline
\end{tabular}
\caption{The largest eigenvalues of $\widetilde{M}^{\HH}_l$}
\label{table:2}
\end{table}



The largest eigenvalues $\lambda^*$ of branching matrice $M^{\text{RWIM}}_l$ and $\widetilde{M}^{\RWIM}_l$ for $l=4,6,8$ are shown in Table~\ref{table:3} and~\ref{table:4} respectively.
When $l = 8$, the largest eigenvalue of $\widetilde{M}^{\RWIM}_l$ is 4.0147, which is less than 4.047. By Lemma \ref{lemma-bm-ssm}, the independent sets of $\lattice^\RWIM$ exhibit SSM. Together with the above discussion, this gives us a (computer-aided) proof of Theorem~\ref{thm-ssm-HH-RWIM}.

\begin{table}[H]
  \centering
\begin{tabular}{|c|c|c|c|}
  \hline
  Max length of Avoiding-cycles & Effect of Order & Number of Types & $\lambda^*$ \\
  \hline
  4 & No & 81 & 4.7273 \\
  \hline
  6 & No & 1003 & 4.6136 \\
  \hline
  8 & No & 13053 & 4.5533 \\
  \hline
\end{tabular}
\caption{The largest eigenvalues of $M_l^{\RWIM}$}
\label{table:3}
\end{table}

\begin{table}[H]
  \centering
\begin{tabular}{|c|c|c|c|}
  \hline
  Max length of Avoiding-cycles & Effect of Order & Number of Types & $\lambda^*$ \\
  \hline
  4 & Yes & 57 & 4.1774 \\
  \hline
  6 & Yes & 603 & 4.0632 \\
  \hline
  8 & Yes & 7238 & 4.0132 \\
  \hline
\end{tabular}
\caption{The largest eigenvalues of $\widetilde{M}^{\RWIM}_l$}
\label{table:4}
\end{table}

\section{Absence of SSM along self-avoiding walks for NAK}
The supertree $T_l^\NAK$ and $\widetilde{T}_l^\NAK$ of $T_\saw(\lattice^\NAK,v)$ generated respectively by branching matrices $M^{\NAK}_l$ and $\widetilde{M}^{\NAK}_l$ can be generated in the same way as stated in Section~\ref{section-supertree-construction}.
The data for these branching matrices is available in our online appendix~\cite{bmdata}.
The largest eigenvalues of $M^{\NAK}_l$ and $\widetilde{M}^{\NAK}_l$ are shown in Table~\ref{table:5} and~\ref{table:6} respectively. Even when $l = 8$, the largest eigenvalue $\lambda^*$ of $\widetilde{M}^{\NAK}_l$ is still far away from what we need in Lemma~\ref{lemma-bm-ssm}, which is $<\gamma=\inf_{x \in [0,+\infty)}\frac{[1+(1+x)^d](1+x)}{dx}\approx 3.917$ for the maximum degree $d+1=8$.

\begin{table}[H]
  \centering
\begin{tabular}{|c|c|c|c|}
  \hline
  Max length of Avoiding-cycles & Effect of Orders & Number of Types & $\lambda^*$ \\
  \hline
  4 & No & 157 & 6.3876  \\
  \hline
  6 & No & 2949 & 6.1894  \\
  \hline
  8 & No & 63205 & 6.0972  \\
  \hline
\end{tabular}
\caption{The largest eigenvalues of $M^{\NAK}_l$}\label{table:5}
\end{table}

\begin{table}[H]
  \centering
\begin{tabular}{|c|c|c|c|}
  \hline
  Max length of Avoiding-cycles & Effect of Orders & Number of Types & $\lambda^*$ \\
  \hline
  4 & Yes & 85 & 4.9883 \\
  \hline
  6 & Yes & 1293 & 4.8275 \\
  \hline
  8 & Yes & 25262 & 4.7587 \\
  \hline
\end{tabular}
\caption{The largest eigenvalues of $\widetilde{M}^{\NAK}_l$}\label{table:6}
\end{table}

We are going to prove that SSM does not hold for a self-avoiding walk tree for $\lattice^\NAK$. We define a homogeneous order of neighbors as follows. For each vertex $v$ in $\lattice^\NAK$, let $\{\text{NW},\text{N},\text{NE},\text{E},\text{SE},\text{S},\text{SW},\text{W}\}$ denote the eight directions leaving vertex $v$. We assign each direction a rank (from 1 to 8) to define the order~$>_v$ of neighbors for $v$, and assume that NW, N and NE are ranked 1,2 and 3 respectively, as shown in Figure~\ref{fig:3}. Let $T_\saw^\NAK=T_\saw(\mathbb{L}^\NAK,v)$ be the self-avoiding walk tree given by this order of neighbors.
Since $\mathbb{L}^\NAK$ is symmetric and the order is homogeneous, the $T_\saw^\NAK$ is isomorphic for all vertices $v$.

\begin{Thm}\label{thm-NAK-no-SSM}
The independent sets of $T_\saw^\NAK$ does not exhibit SSM.
\end{Thm}

\subsection{Subtree construction}
The above theorem is proved by constructing a subtree of $T_\saw^\NAK$ which does not exhibit weak spatial mixing. Then by Proposition~\ref{prop:1}, $T_\saw^\NAK$ does not exhibit SSM.


This subtree of $T_\saw^\NAK$ is constructed by designing a branching matrix in which each generated path corresponds to a self-avoiding walk in $\lattice^\NAK$ and does the necessary truncation for the cycle-closing step as in $T_\saw^\NAK$. This approach is used~\cite{vera} on grid lattice.

\begin{figure}[H]
\centering
  \includegraphics[width=2.1in]{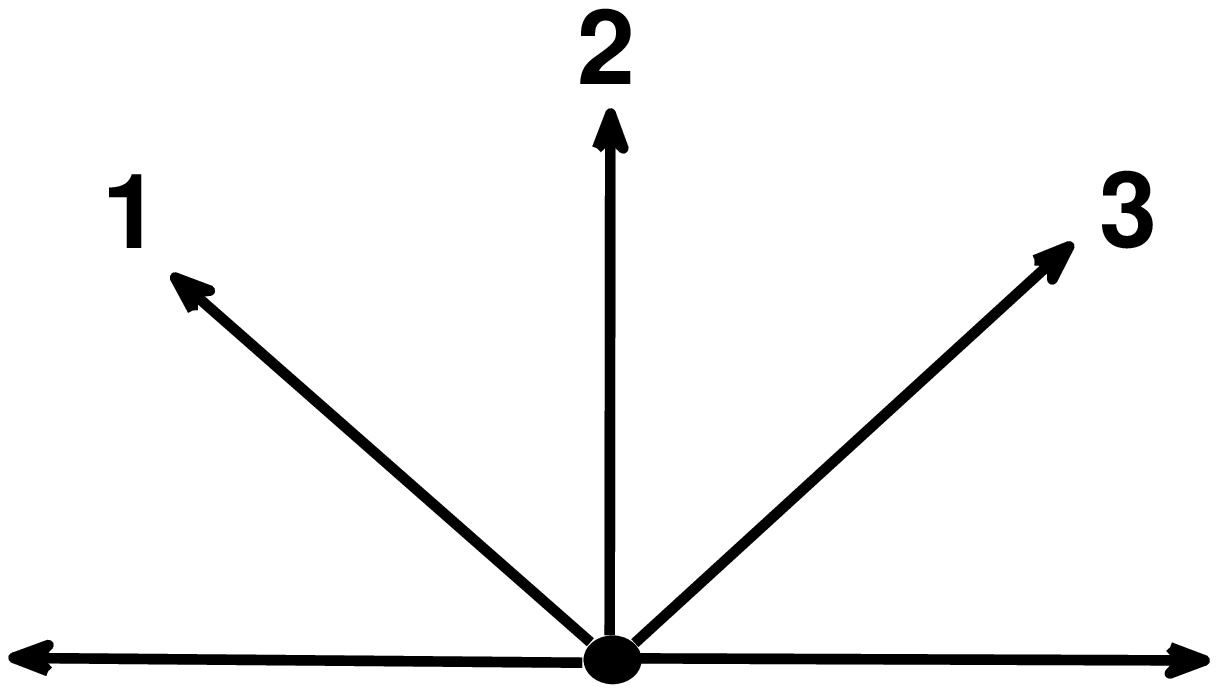}\\
  \caption{Order of neighbors for any vertex in $\lattice^{\NAK}$}\label{fig:3}
\end{figure}

We consider the walks that never go to the three directions SW,S, and SE on south, and never goes back to the direction where it just came from (first going W then E, or first going E then W). Such walks must be self-avoiding since no cycle can be formed. We then further forbid the moves first going NW then E and the moves first going N then E. The remaining walks can be described by a $6\times 6$ branching matrix $M_S$ defined as follows, whose corresponding tree is denoted as $T_{M_S}$.



\begin{align*}
&\,\,\,\,\,\text{O}\,\,\, \text{W}\,\, \text{NW}\,\, \text{N}\,\, \text{NE} \,\, \text{E}\notag\\
M_S =
\begin{matrix}
\text{O}\\
\text{W}\\
\text{NW}\\
\text{N}\\
\text{NE}\\
\text{E}
\end{matrix}
\,\,
&\begin{pmatrix}
0&   1 & 1 & 1 & 1 & 1\\
0&   1 & 1 & 1 & 1 & 0  \\
0&  1 & 1 & 1 & 1 & 0  \\
0&  1 & 1 & 1 & 1 & 0  \\
0&  1 & 1 & 1 & 1 & 1  \\
0&   0 & 1 & 1 & 1 & 1
\end{pmatrix}
\end{align*}

Type O corresponds to the starting point $v$ of the walks. The other types correspond to the five remaining directions $\{\text{W},\text{NW},\text{N},\text{NE},\text{E}\}$, each of which represents the direction of the last step of a path.

\begin{Lem}
The tree $T_{M_S}$ generated by the branching matrix $M_S$ is a subtree of $T_\saw^{\NAK}$.
\end{Lem}

\begin{proof}
Since $M_S$ forbids all walks going to the three directions on south or going back to where it is from, the walks generated by $M_S$ must be self-avoiding. We only need to verify that $M_S$ forbids the walks whose next step closes a cycle from a larger direction than the direction starting the cycle, as in the definition of $T_\saw$ given in Section~\ref{section-SAW}. Specifically, it is sufficient to show that $M_S$ forbids all such walks $v\to\ldots \to w \to x \to\ldots \to u$  that $uw\in E$ and $u >_w x$.
Since the walk never goes south (SW,S and SE), the cycle closing step $u\to w$ must be going to one of the three directions SW, S and SE, which means $u$ is in one of the three directions NW,N and NE from $w$. Since NW,N and NE are ranked 1,2, and 3 respectively, the only bad cases are: (1) $x$ is in NW of $w$ and $u$ is in N or NE of $w$; and (2) $x$ is in N of $w$ and $u$ is in NE of $w$. Neither of cases can happen because $M_S$ forbid any walk first going NW then E, or first going N then E. This shows that $T_{M_S}$ is a subtree of $T_\saw^{\NAK}$.

\end{proof}

The branching matrix $M_S$ can be reduced to a $3 \times 3$ matrix $M_S'$ shown below where the three types corresponds to the respective classes $\{\text{O}\}$, $X=\{\text{W},\text{NW},\text{N},\text{E}\}$, $Y=\{\text{NE}\}$ of old types.
\begin{align}
&\,\,\,\,\,O\,\,\,\, X\,\,\,\, Y\notag\\
M_S' =
\begin{matrix}
O\\
X\\
Y
\end{matrix}
\,\,&\begin{pmatrix}
0 & 4 & 1\\
0 & 3 & 1 \\
0 & 4 & 1
\end{pmatrix}\label{eq:BM-NAK-simple}
\end{align}
It is easy to verify that $M_S$ and $M_S'$ generate the same tree, since for any old types in the same class, the total number of transitions in $M_S$ to all old types in a class is the same, and is given by $M_S'$.

\subsection{Lower bound of correlation decay}

Simulation results show that even weak spatial mixing may not hold for $T_{M_S}$. When the depth of $T_{M_S}$ grows sufficiently large, the value of $|p_{G,v}^+ - p_{G,v}^-|$ approaches to 0.0871958, as shown in Figure \ref{fig:4}, where $p_{G,v}^+$ is the marginal probability of root $v$ being unoccupied when the leaves of $T_{M_S}$ are all fixed to occupied and $p_{G,v}^+$ is the marginal probability of root $v$ being unoccupied when the leaves of $T_{M_S}$ are all fixed to unoccupied. All the leaves are at the same distance from the root $v$.

\begin{figure}[H]
\centering
  \includegraphics[width=4in]{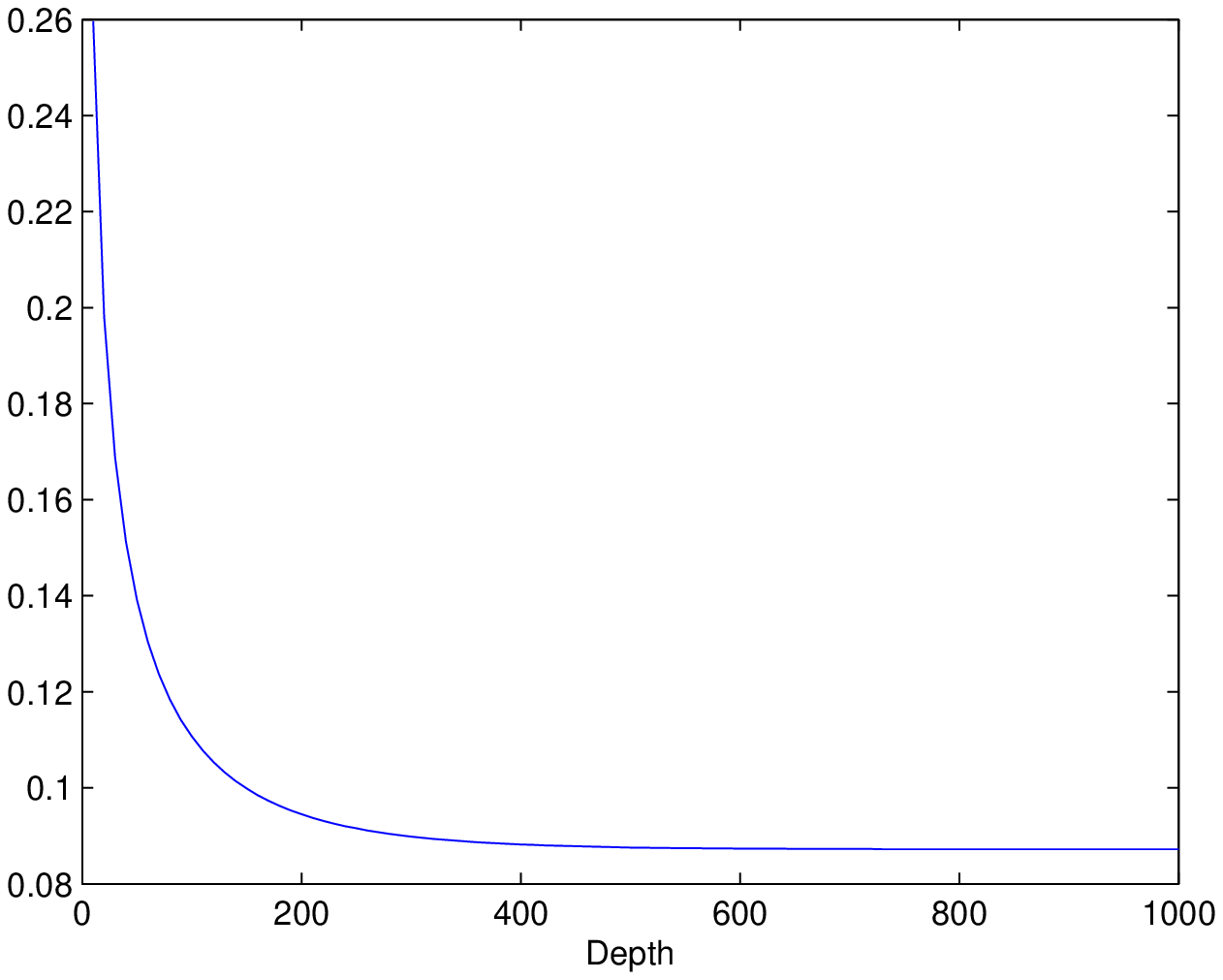}\\
  \caption{The value of $|p_{G,v}^+ - p_{G,v}^-|$ when $T_{M_S}$ is going deep.}\label{fig:4}
\end{figure}

We then rigorously prove that $T_{M_S}$ indeed does not exhibit WSM. Since $T_{M_S}$ is a subtree of $T_\saw^\NAK$, this implies that the self-avoiding walk tree $T_\saw^\NAK$ of $\lattice^\NAK$ does not exhibit SSM, proving Theorem~\ref{thm-NAK-no-SSM}.
\begin{Lem}
WSM does not hold on $T_{M_S}$.
\end{Lem}

\begin{proof}
Let $S$ be the cutset of $T_{M_S}$ which contains all the vertices at distance $\ell$ from the root. Supposed we fix the vertices in $S$ to be all unoccupied or to be all occupied, by induction it is easy to see that the marginal distributions at all vertices of the same type at the same level of the tree above $S$ are the same.
Recall that $T_{M_S}$ is captured by the simplified branching matrix $M_S'$ defined in~\eqref{eq:BM-NAK-simple} of three types $O,X,Y$. Let $x^{(t)}$ (and $y^{(t)}$) denote the ratio between probabilities of being occupied and unoccupied at a vertex of type X (and type Y) at distance $t$ from the boundary $S$.
Then by equation (\ref{branching-recursion}), we have the following recursion:
\begin{align*}
x^{(t)}=F_1\left(x^{(t-1)},y^{(t-1)}\right)
&=\frac{1}{\left(1+x^{(t-1)}\right)^3\left(1+y^{(t-1)}\right)},\\
y^{(t)}=F_2\left(x^{(t-1)},y^{(t-1)}\right)
&=\frac{1}{\left(1+x^{(t-1)}\right)^4\left(1+y^{(t-1)}\right)},
\end{align*}
and $x^{(0)}=y^{(0)}=0$ or $x^{(0)}=y^{(0)}=\infty$ depending on whether the vertices in $S$ are fixed to be unoccupied or occupied. The weak spatial mixing on $T_{M_S}$ holds only if the system converges as $t\to\infty$.

The Jacobian matrix is given by
\begin{equation*}
\mathcal{J} =
\begin{pmatrix}
|\frac{\partial F_1}{\partial x}| & |\frac{\partial F_1}{\partial y}| \\
|\frac{\partial F_2}{\partial x}| & |\frac{\partial F_2}{\partial y}|
\end{pmatrix}
=
\begin{pmatrix}
\frac{3}{(1+x)^4(1+y)} & \frac{1}{(1+x)^3(1+y)^2} \\
\frac{4}{(1+x)^5(1+y)} & \frac{1}{(1+x)^4(1+y)^2}
\end{pmatrix}.
\end{equation*}
Suppose that $(\hat{x},\hat{y})$ be the unique nonnegative fixed point of the system satisfying that $\hat{x}=F_1(\hat{x},\hat{y})$ and $\hat{y}=F_2(\hat{x},\hat{y})$. Let $\widehat{\mathcal{J}}=\left.\mathcal{J}\right|_{x=\hat{x},y=\hat{y}}$ be the Jacobian matrix at the fixed point. We are going to show that there exists a nonnegative fixed point $(\hat{x},\hat{y})$ such that $\widehat{\mathcal{J}}$ has an eigenvalue greater than 1. Then due to~\cite{robinson}, the function around the the fixed point is repelling and hence it is impossible to converge to this unique fixed point.


We define the functions $\delta_1(x,y) = \frac{1}{(1+x)^3(1+y)} - x$ and $\delta_2(x,y) = \frac{1}{(1+x)^4(1+y)} - y$. Then $(\hat{x}, \hat{y})$ is the real positive solution of the equations $\delta_1(x,y) = 0$ and $\delta_2(x,y) = 0$. Then we have
\begin{equation*}
\frac{\hat{x}}{1+\hat{x}} - \frac{1}{\hat{x}(1+\hat{x})^3} + 1 = 0,
\end{equation*}
\begin{equation*}
\sqrt[4]{\frac{\hat{y}^3}{1+\hat{y}}} - \sqrt[4]{\frac{1}{\hat{y}(1+\hat{y})}} + 1 = 0.
\end{equation*}

We define that $\hat{F}_1(x) = \frac{x}{1+x} - \frac{1}{x(1+x)^3} + 1$ and $\hat{F}_2(y) = \sqrt[4]{\frac{y^3}{1+y}} - \sqrt[4]{\frac{1}{y(1+y)}} + 1$. It holds that
\begin{equation*}
\hat{F}'_1(x) = \frac{1}{(1+x)^2} + \frac{1+4x}{x^2(1+x)^4} > 0,
\end{equation*}
\begin{equation*}
\hat{F}'_2(y) = \frac{1}{4} [y(1+y)]^{-\frac{5}{4}} (1+5y+2y^2) > 0,
\end{equation*}
for $x,y> 0$.
Therefore $\hat{F}_1(x)$ and $\hat{F}_2(y)$ are increasing for all positive $x, y$.

When $\tilde{x} = 0.3356$ and $\tilde{y} = 0.2513$, $\hat{F}_1(\tilde{x}) = 5.8531 \times 10^{-4} > 0$, $\hat{F}_2(\tilde{y}) = 1.8569 \times 10^{-4} > 0$. Hence we can conclude that $\tilde{x} > \hat{x}$ and $\tilde{y} > \hat{y}$. Since every entry of Jacobian matrix $\mathcal{J}$ is decreasing in both $x$ and $y$, then $\widetilde{\mathcal{J}}$ is entry-wise smaller than $\widehat{\mathcal{J}}$, where $\widetilde{\mathcal{J}}$ is the Jacobian matrix at point $(\tilde{x}, \tilde{y})$, given by
\begin{equation*}
\widetilde{\mathcal{J}} =
\begin{pmatrix}
0.7534 & 0.2681 \\
0.7522 & 0.2007
\end{pmatrix}.
\end{equation*}

The maximum eigenvalue of $\widetilde{\mathcal{J}}$ is $\lambda^*(\widetilde{\mathcal{J}}) = 1.0044 > 1$. Since both $\widetilde{\mathcal{J}}$ and $\widehat{\mathcal{J}}$ are positive matrices and $\widetilde{\mathcal{J}}<\widehat{\mathcal{J}}$ entry-wisely, we can conclude that $\lambda^*(\widehat{\mathcal{J}}) > \lambda^*(\widetilde{\mathcal{J}}) >1$. This proves the lemma.

\end{proof}


\section{Conclusions}
In this paper, we give PTAS for computing the capacities of two-dimensional codes with constraints HH and RWIM using strong spatial mixing. We also show that the capacity of two-dimensional code with constraint NAK may not be approximated efficiently this method. 
An important open direction is to generalize this approach to other constraints and higher dimensions.

\paragraph{Acknowledgment.} We are deeply grateful to Mordecai Golin for many helpful discussions and pointing us to the problem of computing hard-square entropy.





\end{document}